\documentclass[12pt]{iopart}
\usepackage{iopams}
\usepackage{amscd}
\usepackage{color}

\newtheorem{lem}{Lemma}
\newtheorem{prop}{Proposition}
\newenvironment{proof}{\textbf{Proof}:}{\hfill\rule{4pt}{8pt}}

\allowbreak

\begin{document}

\title[Conservation laws and symmetries of Hunter-Saxton equation]{Conservation laws and symmetries of Hunter-Saxton equation: revisited}

\author{Kai Tian$^1$\footnote{Corresponding author} \& Q. P. Liu$^2$}

\address{Department of Mathematics, China University of Mining and Technology, Beijing, 100083, P. R. China}
\ead{\mailto{$^1$tiankai@lsec.cc.ac.cn}, \mailto{$^2$qpl@cumtb.edu.cn}}
\vspace{10pt}
\begin{indented}
\item[] \today
\end{indented}

\begin{abstract}
Through a reciprocal transformation $\mathcal{T}_0$ induced by  the conservation law $\partial_t(u_x^2) = \partial_x(2uu_x^2)$, the Hunter-Saxton (HS) equation $u_{xt} = 2uu_{2x} + u_x^2$ is shown to possess conserved densities involving arbitrary smooth functions, which have their roots in infinitesimal symmetries of $w_t = w^2$, the counterpart of the HS equation under $\mathcal{T}_0$. Hierarchies of commuting symmetries of the HS equation are studied under appropriate changes of variables initiated by $\mathcal{T}_0$, and two of these are linearized while the other is identical to the hierarchy of commuting symmetries admitted by the potential modified Korteweg-de Vries equation. A fifth order symmetry of the HS equation is endowed with a sixth order hereditary recursion operator by its connection with the Fordy-Gibbons equation. These results reveal the origin for the rich and remarkable structures of the HS equation and partially answer the questions raised by Wang [\NL {\bf 23}(2010) 2009].
\end{abstract}

\vspace{2pc}

%\submitto{\NL}
%\maketitle

\section{Introduction}
The following equation, known as the Hunter-Saxton (HS) equation,
\begin{equation}\label{hs}
u_{xt} = 2uu_{2x}+u_x^2,
\end{equation}
where subscripts denote partial derivatives and $u_{jx}$ stands for ${\partial^j u}/{\partial x^j}$, was proposed by Hunter and Saxton to describe weakly nonlinear orientation waves in a massive nematic liquid crystal direct field \cite{saxton}. This equation has been studied extensively  and various interesting properties have been established from different points of view \cite{beals,constantin,zheng1,zheng2,zheng3,khe,lenells}. In particular, it was shown that the HS equation has a remarkable property \cite{zheng1}, namely it can be derived from two different variational principals and one of them is the high-frequency limit of the variational principle for the Camassa-Holm equation, which is completely integrable \cite{cama}.  The bi-variational property implies that the HS equation (\ref{hs}) is a completely integrable bi-Hamiltonian system and may  be embedded in the Harry Dym hierarchy \cite{zheng1}. The evolutionary form of HS equation
\begin{equation}\label{hs:nonlocal}
u_t = 2uu_x - \partial_x^{-1} u_x^2,
\end{equation}
where the integrating operator $\partial_x^{-1}$ is defined as the inverse of $\partial_x$ such that $\partial_x\partial_x^{-1} = \partial_x^{-1}\partial_x =1$, was interpreted as a governing equation of shortwave perturbation in one-dimensional relaxing medium \cite{golenia}. In addition to its physical importance, the HS equation is also remarkable in geometric view as the geodesic flow equation in a homogeneous space with right-invariant metrics \cite{khe,lenells}.

Recently, certain extraordinary algebraic and geometric properties of the HS equation were established, such as  infinitely many $(x,t)$-independent conserved densities with free parameters \cite[Theorem 2]{wang}, and three non-equivalent hereditary recursion operators \cite[Theorem 3]{wang}, which yield three hierarchies of higher order infinitesimal symmetries (abbreviated to symmetries) independent of $x$ and $t$. Symmetries belonging to the same hierarchy are commutative with each other, while those generated by different recursion operators non-commutative. This phenomenon is contrary to the common belief about $(x,t)$-independent conserved densities (or symmetreis) in the theory of integrable evolution equations \cite{fokas1,olverbook}, and distinguishes the HS equation from typical completely integrable equations, represented by the Korteweg-de Vries (KdV) equation \cite{miura}. As posed at the end of the paper  \cite{wang}, the fundamental and natural questions are:  why does the HS equation possess so many conserved densities and symmetries, and where do they come from?

This paper aims to answer these questions, and provides reasonable explanations about exceptional features of the HS equation. Two general philosophies in contents of symmetry methods \cite{olverbook} may lead us to the destination. On the one hand, linear equations usually exhibit rich symmetry structures, and as their disguises, linearizable equations inherit abundant symmetries from them. On the other hand, there are enormous symmetries for flow equations of basic point transformations, for instance  the flow equation of $x$-translation $u_t=u_x$ takes all smooth functions $F(u,u_x,u_{2x},\cdots, u_{nx})$ as its symmetries, while the flow equation of $u$-translation $u_t=1$ admits symmetries like $F(u_x,u_{2x},\cdots, u_{nx})$, where $F$ stands for an arbitrary smooth function of its arguments.

It is noted that via reciprocal transformations the HS equation \eref{hs} is brought  to either the Liouville equation \cite{dai}
\begin{equation}\label{liou}
\phi_{yt} = 2\exp(\phi),
\end{equation}
or
\begin{equation}\label{hsode}
w_t = w^2.
\end{equation}
an equation, which can be solvable by quadrature \cite{pavlov}.
Since the Liouville equation is known as  linearizable or C-integrable,  its symmetries are very rich in the sense of involving arbitrary functions \cite{zhiber}. Furthermore, \Eref{hsode} is equivalent to the flow equation of $v$-translation $v_t = 1$ by assuming $w=-v^{-1}$, so its symmetries are also enormous. In fact, general solutions to the HS equation \eref{hs} was constructed from solutions to \Eref{hsode}. We will take \Eref{hsode} as the key to reveal secretes of novel properties of the HS equation.

This paper is organized as follows. In the next section, various reciprocal transformations, implied by  different conservation laws, are shown to convert the HS equation \eref{hs} to \Eref{hsode}, and one of them (denoted as $\mathcal{T}_0$), generated by a simple conservation law, will play a crucial role in constructing conservation laws and symmetries of the HS equation. Induced by this universal reciprocal transformation, a direct relation is built up between the linearized symmetry equation of \Eref{hsode} and a class of conservation laws of the HS equation, and thus in section \ref{sec:3} two classes of conservation laws (Propositions \ref{pro:1} and \ref{pro:2}) with arbitrary functions are discovered for the HS equation. In section \ref{sec:4}, we consider some hierarchies \cite{wang} of commuting symmetries of the HS equation by transforming them to familiar models through appropriate changes of variables, and in a similar way establish a hereditary recursion operator for a fifth order symmetries of the HS equation.  As stated in the last section, these results give a reasonable explanation about the unusual algebraic properties of the HS equation, and partially answer questions raised by Wang \cite{wang}.

\section{Reciprocal relations between the HS equation \eref{hs} and \Eref{hsode}} \label{sec:2}
We first review some  facts on the reciprocal transformations for the HS equation, reported in \cite{dai,pavlov}.  It was shown that the HS equation \eref{hs} is changed to either the Liouville equation \eref{liou}  or \Eref{hsode}, by a reciprocal transformation defined by the following conservation law
\begin{equation}\label{conser:a}
\partial_t\Big(u_{2x}^{1/2}\Big) = \partial_x\Big(2uu_{2x}^{1/2}\Big).
\end{equation}

In fact, certified by the conservation law \eref{conser:a}, a new variable $z$ is defined as
\[
\textmd{d}z = u_{2x}^{1/2} \textmd{d}x + 2uu_{2x}^{1/2} \textmd{d}t,
\]
and it serves as the spatial variable instead of $x$, and triggers an invertible transformation
\[
\mathcal{T}_{1/2}: (x,t,u)\rightarrow (z,t,w),
\]
where $w(z(x,t),t)=u_x$, which relates the HS equation \eref{hs} to Equations \eref{liou} and \eref{hsode}.

To formulate the HS equation \eref{hs} in new variables, one should notice
\begin{equation*}
\left(\eqalign{
\partial_x \\
\partial_t}\right) \Rightarrow \left(\eqalign{
 q\partial_z  \\
 \partial_t + 2uq\partial_z   }\right),\quad \left(q(z(x,t),t)= u_{2x}^{1/2}\right)
\end{equation*}
which is inferred from the chain rule of differentiation. With above transformation formulae, it is easy to check that the HS equation \eref{hs} is converted to \Eref{hsode}, while the conservation law itself \eref{conser:a} becomes
\begin{equation*}\label{eq:2:1}
q_t = 2q^2u_z,
\end{equation*}
which may be rewritten as
\begin{equation}\label{eq:2:2}
\partial_t\Big(\ln q\Big) = 2qu_z = 2u_x = 2w.
\end{equation}
Indicated by definitions of $w$ and $q$, we have a relation between them, given by
\[
q^2 = u_{2x} = q w_z,
\]
which means $q=w_z$. Then differentiating both sides of \Eref{eq:2:2} with respect to $z$ yields
\[
\partial_t\partial_z\Big(\ln q\Big)= 2w_z =2q,
\]
which is identical to the Liouville equation \eref{liou} assuming $q = \exp(\phi)$.

It should be pointed out that {\it the reciprocal transformation $\mathcal{T}_{1/2}$ can not be applied to all commuting flows of the HS equation}, since the conserved density $u_{2x}^{1/2}$ are not always involutory with others under the Poisson bracket defined by a nonlocal Hamiltonian operator $\partial_x^{-1}$. This drawback limits the implications of $\mathcal{T}_{1/2}$ and prevents it from clearly unfolding algebraic and geometric structures of the HS equation. An universal transformation for all commuting flows is expected, and we rest our hope on more simple conservation laws.

As a byproduct, many conservation laws different from \eref{conser:a} can also be utilized to generate reciprocal transformations relating the HS equation \eref{hs} with \Eref{hsode}. Indeed, it is direct to check that the HS equation \eref{hs} admits a conservation law
\begin{equation}\label{conser:b}
\partial_t\Big(u_x^2F(r_2)\Big) = \partial_x\Big(2uu_x^2F(r_2)\Big),\quad(r_2\equiv u_x^{-4}u_{2x})
\end{equation}
where $F$ is an arbitrary smooth function of its argument. Basing on it, we define the change of variables
\[
\mathcal{T}_F:(x,t,u)\rightarrow (y,t,w)
\]
where $w(y(x,t),t) = u_x$ and $y$ is given by
\begin{equation*} %\label{def:y}
\textmd{d}y = u_x^2F(r_2)\textmd{d}x + 2uu_x^2F(r_2)\textmd{d}t.
\end{equation*}
Consequently, derivatives with respect to independent variables are related through the rule
\begin{equation*}
\left(\eqalign{
	\partial_x \\
	\partial_t}\right) \Rightarrow \left(\eqalign{
	q\partial_y  \\
	\partial_t + 2uq\partial_y   }\right),\quad \Big(q(y(x,t),t)=u_x^2F(r_2)\Big).
\end{equation*}

Under this new transformation $\mathcal{T}_F$, the HS equation \eref{hs} is also converted to \Eref{hsode}, while the conservation law \eref{conser:b} to
\begin{equation}\label{eq:2:3}
q_t = 2q^2u_y,
\end{equation}
and the relation between $w$ and $q$ in new independent variable $y$ reads as
\begin{equation*}
q=w^2F\left(qw^{-4}w_y\right).
\end{equation*}

We also notice that \Eref{eq:2:3} can be reformulated as a conservation law
\begin{equation*} %\label{conser:bn}
\partial_t\Big(q^{-1}\Big) = \partial_y\Big(-2u\Big),
\end{equation*}
which allows us to define the inverse of $\mathcal{T}_F$ and recover the HS equation \eref{hs} from \Eref{hsode}.

Obviously, if one takes $F(r_2)=r_2^{1/2}$, then all the things return back to the previous results produced by the reciprocal transformation $\mathcal{T}_{1/2}$. However   a simpler choice for $F$ is feasible, namely the  conservation law \eref{conser:b} with $F(r_2) = 1$ gives
\begin{equation*}\label{conser:c}
\partial_t\Big(u_x^2\Big) = \partial_x\Big(2uu_x^2\Big),
\end{equation*}
which results in the transformation
\begin{equation*}
	\mathcal{T}_0:(x,t,u)\rightarrow (y,t,w),
\end{equation*}
where $w(y(x,t),t) = u_x$ and $y$ is defined by
\[
\textmd{d}y = u_x^2\textmd{d}x + 2uu_x^2\textmd{d}t.
\]

The conserved density $u_x^2$ is involutory with all others independent of $x$ under the Poisson bracket defined by $\partial_x^{-1}$ since the Hamiltonian flow generated by $u_x^2/2$ is just the $x$-translation. Hence {\it the reciprocal transformation $\mathcal{T}_0$ can be applied to any commuting flow of the HS equation.} With its help, we manage to get a clear picture of the existence of  abundant conservation laws, symmetries and recursion operators of the HS equation  and this is the content of the subsequent sections.

\section{Conserved densities with arbitrary functions}\label{sec:3}
As an interesting observation, we notice that by means of the inverse of $\mathcal{T}_0$ the linearized symmetry equation of \Eref{hsode}, namely,
\begin{equation}\label{hsodesym}
\left.\Big(\partial_t F(w,w_y,w_{2y}\cdots,w_{ny}) - 2w F(w,w_y,w_{2y}\cdots,w_{ny})\Big)\right|_{w_t = w^2} = 0
\end{equation}
is converted to
\begin{equation}\label{conser:tp}
\left.\left(\partial_t\Big(\mathcal{T}_0^{-1}(F)\Big) - \partial_x\Big(2u\mathcal{T}_0^{-1}(F)\Big)\right)\right|_{u_{xt} = 2uu_{2x}+u_x^2} = 0,
\end{equation}
where the image of $F$ under $\mathcal{T}_0^{-1}$ is given by
\begin{equation*}
\mathcal{T}_0^{-1}(F) = F\Big(u_x,u_x^{-2}u_{2x},(u_x^{-2}\partial_x)^2(u_x),\cdots,(u_x^{-2}\partial_x)^n(u_x)\Big).
\end{equation*}
In contents of symmetry methods \cite{olverbook}, any function $F(w,w_y,w_{2y}\cdots,w_{ny})$ is referred to as an infinitesimal symmetry of \Eref{hsode} if it satisfies \Eref{hsodesym}, while the function $\mathcal{T}_0^{-1}(F)$ would be a conserved density of the HS equation if \Eref{conser:tp} holds.  The intimate relation between Equations \eref{hsodesym} and \eref{conser:tp} enlightens us to produce conserved densities for the HS equation by applying $\mathcal{T}_0^{-1}$ to infinitesimal symmetries of \Eref{hsode}.

Assuming $w=-v^{-1}$, \Eref{hsode} is further brought to
\begin{equation*}
v_t=1
\end{equation*}
for which  any smooth function $G(v_y,v_{2y},\cdots,v_{ny})$, not explicitly depending on $v$, is a symmetry. Coordinating with the transformation $v=-w^{-1}$, any smooth function like
\begin{equation*}
w^2G\left(\partial_y(-w^{-1}),\partial_y^2(-w^{-1}),\cdots,\partial_y^n(-w^{-1})\right)
\end{equation*}
solves \Eref{hsodesym}, thus does supply  a symmetry of \Eref{hsode}.  Applying $\mathcal{T}_0^{-1}$ to it yields a conserved density of the HS equation, given by
\begin{equation*}
u_x^2G(r_2,r_3,\cdots,r_{n+1}),
\end{equation*}
which provides the conservation law
\begin{equation*}
\partial_t\Big(u_x^2G(r_2,r_3,\cdots,r_{n+1})\Big) = \partial_x\Big(2uu_x^2G(r_2,r_3,\cdots,r_{n+1})\Big),
\end{equation*}
where $r_{k+1} \equiv \mathcal{T}_0^{-1}(\partial_y^k(-w^{-1}))  = (u_x^{-2}\partial_x)^k\left(-u_x^{-1}\right), ~k=1,2,\cdots,n$.

For the HS equation, it is remarked that the conserved densities with free parameters, constructed in a recursive approach \cite{wang} and formulated in our notations as
\begin{equation*}
u_x^2 r_2^{\alpha_1}r_3^{\alpha_2}\cdots r_{n+1}^{\alpha_n}\quad (\alpha_1,\alpha_2,\cdots,\alpha_n\mbox{ are arbitrary constants.})
\end{equation*}
can be understood as transparent corollaries of our result. The freedom brought in by free parameters  is attributed to arbitrary smooth functions involved in symmetries of \Eref{hsode}, even more essentially, $v_t=1$.

Accompanying with a big family of conserved densities such that
\begin{equation}\label{law:1}
\partial_t\Big(F(u_x,u_{2x},\cdots,u_{nx})\Big) = \partial_x\Big(2uF(u_x,u_{2x},\cdots,u_{nx})\Big),
\end{equation}
some questions naturally arise, including
\begin{enumerate}
	\item Are there conserved densities other than $u_x^2G(r_2,r_3,\cdots,r_{n+1})$ enabling \Eref{law:1} to be conservations laws for the HS equation?
	\item Is it possible to establish conservation laws of new type different form \Eref{law:1} for the HS equation?
\end{enumerate}
These questions  will be answered by the following propositions.

\begin{lem}\label{lem:1}
	Let $r_1 = -u_x^{-1}$ and $r_{k+1} = (u_x^{-2}\partial_x)^k\left(-u_x^{-1}\right)~(k=1,2,\cdots)$, then
	\begin{equation*}
	\partial_t r_1 = 2u(\partial_xr_1)+1\quad\mbox{and}\quad \partial_t r_k = 2u(\partial_xr_k)\quad(k=2,3,\cdots)
	\end{equation*}
	when the HS equation \eref{hs} holds.
\end{lem}
\begin{proof}
By direct calculation, the evolution of $r_1$ is obtained, and then with aid of the recursive relation $r_{k+1} = u_x^{-2}(\partial_xr_k)$,  one may verify evolution equations for other $r_k$'s by induction.
\end{proof}

\begin{prop}\label{pro:1}
For the HS equation, $F(u_x,u_{2x},\cdots,u_{nx})$ is a conserved density solving \Eref{law:1}  if and only if
\begin{equation*}
F(u_x,u_{2x},\cdots,u_{nx}) = u_x^2G(r_2,r_3,\cdots,r_n),
\end{equation*}
where $r_k$'s are defined in Lemma \ref{lem:1} and  $G$ is an arbitrary smooth function of its arguments.
\end{prop}
\begin{proof}
The sufficiency follows from the previous paragraphs, so we just prove the necessity by solving all conserved densities  from the assumed conservation law \eref{law:1}.

To this end, we first consider the following invertible change of coordinates
\begin{equation*}
\Gamma: (u_x,u_{2x},\cdots,u_{nx})\mapsto(r_1,r_2,\cdots,r_n),
\end{equation*}
which is well-defined since $\partial r_k/\partial u_{kx} = u_x^{-2k}~(k=1,2,\cdots)$.  Then, the unknown function $F(u_x,u_{2x},\cdots,u_{nx})$ in the new coordinates is denoted as $\hat{F}(r_1,r_2,\cdots,r_n)$  and in this way the conservation law \eref{law:1} yields
\begin{equation*}
\eqalign{
0 &= \partial_t\Big(\hat{F}(r_1,r_2,\cdots,r_n)\Big) - \partial_x\Big(2u\hat{F}(r_1,r_2,\cdots,r_n)\Big) \\
&= \sum_{k=1}^n (\partial_tr_k)\frac{\partial \hat{F}}{\partial r_k} - 2u\sum_{k=1}^n (\partial_xr_k)\frac{\partial \hat{F}}{\partial r_k} + 2r_1^{-1} \hat{F} .
}
\end{equation*}
Taking Lemma \ref{lem:1} into account, we have
\begin{equation*}
\eqalign{
0 &= \sum_{k=1}^n 2u(\partial_xr_k)\frac{\partial \hat{F}}{\partial r_k} + \frac{\partial \hat{F}}{\partial r_1} - 2u\sum_{k=1}^n\frac{\partial \hat{F}}{\partial r_k}(\partial_xr_k) + 2r_1^{-1} \hat{F} \\
&= \frac{\partial \hat{F}}{\partial r_1} + 2r_1^{-1} \hat{F}.
}
\end{equation*}
Solving it by quadrature, we obtain
\begin{equation*}
\hat{F}(r_1,r_2,\cdots,r_n) = r_1^{-2} G(r_2,r_3,\cdots,r_n),
\end{equation*}
which implies $F(u_x,u_{2x},\cdots,u_{nx}) = u_x^2G(r_2,r_3,\cdots,r_n)$, where $G$ stands for an arbitrary smooth function of its arguments.
\end{proof}

To get conservation laws of new type for the HS equation, let us consider a smooth function
\begin{equation*}
H(u_x,u_{2x},\cdots,u_{nx}) \equiv r_1^{-2} P(r_1,r_2,\cdots,r_n),
\end{equation*}
which evolves on solutions to the HS equation as
\begin{equation*}
\eqalign{
\partial_t H & = -2r_1^{-3}(\partial_t r_1)P + r_1^{-2}\sum_{k=1}^n (\partial_t r_k)\frac{\partial P}{\partial r_k} \\
&= -2r_1^{-3}\Big(2u(\partial_x r_1)+1\Big)P + r_1^{-2} \left(\frac{\partial P}{\partial r_1} + \sum_{k=1}^n 2u(\partial_x r_k)\frac{\partial P}{\partial r_k}\right) \\
&= \partial_x \Big(2ur_1^{-2} P\Big) + r_1^{-2}\frac{\partial P}{\partial r_1}.
}
\end{equation*}
Therefore, if there exists $Q(r_1,r_2,\cdots,r_{n-1})$ such that
\begin{equation*}
\frac{\partial P}{\partial r_1} = r_1^2 (\partial_x Q),
\end{equation*}
then $H(u_x,u_{2x},\cdots,u_{nx}) = r_1^{-2} P(r_1,r_2,\cdots,r_n)$ would be a conserved density such that a new conservation law
\begin{equation*}
\partial_t H = \partial_x\Big(2uH + Q\Big).
\end{equation*}

Summarizing above discussion, we have
\begin{prop}
	For any smooth function $Q(r_1,r_2,\cdots,r_{n-1})~(n\geq 2)$, let
	\begin{equation*}\label{pro:2}
	P(r_1,r_2,\cdots,r_n) = \int r_1^2 (\partial_x Q) \textmd{d}r_1,
	\end{equation*}
	then $H(u_x,u_{2x},\cdots,u_{nx}) = r_1^{-2} P(r_1,r_2,\cdots,r_n)$ is a conserved density of the HS equation such that
	\begin{equation}\label{law:2}
	\partial_t H(u_x,u_{2x},\cdots,u_{nx}) = \partial_x\Big(2uH(u_x,u_{2x},\cdots,u_{nx}) + Q(r_1,r_2,\cdots,r_{n-1})\Big).
	\end{equation}
\end{prop}

As an example, we will deduce a conserved density from a given smooth function $Q(r_2)$.  Rewriting the recursive relation among $r_k$'s as
\begin{equation*}
\partial_x r_k = u_x^2r_{k+1} = r_1^{-2}r_{k+1}~(k=1,2,\cdots),
\end{equation*}
hence differentiating $Q(r_2)$ with respect to $x$ yields
\begin{equation*}
\partial_xQ(r_2) = (\partial_x r_2)Q'(r_2) = r_1^{-2} r_3Q'(r_2),
\end{equation*}
where $Q'(r_2)={dQ}/{dr_2}$, namely the prime denotes derivative with respect to the relevant argument. Subsequently, we have
\begin{equation*}
P(r_1,r_2,r_3) = \int r_1^2 (\partial_x Q) \textmd{d}r_1 = \int r_3Q'(r_2) \textmd{d}r_1 = r_1r_3Q'(r_2) + T(r_2,r_3),
\end{equation*}
where $T(r_2,r_3)$ is an arbitrary smooth function produced by integration.  Therefore, a conserved density of the HS equation is given by
\begin{equation*}
\eqalign{
H(u_x,u_{2x},u_{3x}) &= r_1^{-2}P(r_1,r_2,r_3) = r_1^{-1}r_3Q'(r_2) + r_1^{-2}T(r_2,r_3) \\
&= -u_x\Big(u_x^{-6}u_{3x}-4 u_x^{-7}u_{2x}^2\Big)Q'(r_2) + u_x^2 T(r_2,r_3) \\
&= - u_x^{-5}u_{3x}Q'(r_2)  + 4u_x^{-6}u_{2x}^2Q'(r_2)+ u_x^2 T(r_2,r_3) \\
&= - u_x^{-5}u_{3x}Q'(r_2) + u_x^2 \hat{T}(r_2,r_3),
}
\end{equation*}
where $ \hat{T}(r_2,r_3) = 4r_2^2Q'(r_2)+T(r_2,r_3)$, and satisfies
\begin{equation*}
\partial_t H(u_x,u_{2x},u_{3x}) = \partial_x\Big(2uH(u_x,u_{2x},u_{3x}) + Q(r_2)\Big).
\end{equation*}
Because $u_x^2 \hat{T}(r_2,r_3)$ is also a conserved density according to Proposition \ref{pro:1}, and satisfies
\begin{equation*}
 \partial_t \Big(u_x^2 \hat{T}(r_2,r_3)\Big) = \partial_x\Big(2uu_x^2 \hat{T}(r_2,r_3)\Big),
\end{equation*}
the conserved density $H(u_x,u_{2x},u_{3x})$ may be replaced by an reduced one $- u_x^{-5}u_{3x}Q'(r_2)$,  with the corresponding conservation law
\begin{equation*}
\partial_t \Big(- u_x^{-5}u_{3x}Q'(r_2)\Big) = \partial_x \Big(-2uu_x^{-5}u_{3x}Q'(r_2) + Q(r_2)\Big),\quad \left(r_2=u_x^{-4}u_{2x}\right).
\end{equation*}

Both propositions are very fruitful to produce conserved densities with arbitrary functions,   however the conserved densities constructed above do not exhaust all possibilities,  for instance the HS equation may also have nonlocal conserved densities  \cite{popo}.  As supplements to this section, nonlocal conserved densities with arbitrary function will be discussed in appendix.

\section{Symmetries and recursion operators demystified }\label{sec:4}
The evolutionary form \eref{hs:nonlocal} of HS equation is formulated as Hamiltonian systems in two different ways
\begin{equation*}
u_t = 2uu_x - \partial_x^{-1}u_x^2
= \partial_x^{-1} \delta_u (-uu_x^2) = \Big(u_x\partial_x^{-2} - \partial_x^{-2}u_x\Big)\delta_u(-u_x^2),
\end{equation*}
where  $\delta_u$ denotes the Euler derivative with respect to $u$. Thus, the pair of compatible Hamiltonian operators $\partial_x^{-1}$ and $(u_x\partial_x^{-2} - \partial_x^{-2}u_x)$  combines into a hereditary recursion operator
\begin{equation*}
\mathcal{R}(u) = \Big(u_x\partial_x^{-2} - \partial_x^{-2}u_x\Big)\partial_x
\end{equation*}
for the HS equation. It is well known that a typical integrable system often admits a hereditary recursion operator. However, for the HS equation, what is interesting is that it possesses new  recursion operators other than $\mathcal{R}(u)$.
%If the story stops here, then there is nothing exciting. However something surprising did happen to the HS equation \eref{hs:nonlocal}.
Indeed, by classifying a kind of anti-symmetric operators constituting Hamiltonian pairs with $\partial_x^{-1}$, Wang isolated three hereditary (or Nijenhuis) recursion operators for the HS equation \eref{hs:nonlocal}, namely
\begin{equation*}
\eqalign{
\mathcal{R}^{(1)}(u) &= \Big(2u_{2x}^{-1}\partial_x + 2\partial_x u_{2x}^{-1} - u_{2x}^{-3/2}u_{3x}\partial_x^{-1}u_{2x}^{-3/2}u_{3x}\Big)\partial_x, \\
\mathcal{R}^{(2)}(u) &= \Big(u_x^{4}u_{2x}^{-2}\partial_x + \partial_x u_x^{4}u_{2x}^{-2} - 8u_x\partial_x^{-1}u_x\Big)\partial_x, \\
\mathcal{R}^{(3)}(u) &= \Big(u_x^{-4}\partial_x + \partial_x u_x^{-4} - 4 h \partial_x^{-1} u_x - 4 u_x\partial_x^{-1} h\Big)\partial_x, ~~ (h\equiv u_x^{-6}u_{3x}-3 u_x^{-7}u_{2x}^2),
}
\end{equation*}
and clearly showed $\mathcal{R}^{(1)}(u)$ is indeed the inverse of $\mathcal{R}(u)$ so not a truly new recursion operator, but the other two are new  \cite{wang}. As implications of these hereditary recursion operators, three hierarchies of commuting symmetries for the HS equation could be formulated as
\begin{equation*}
\eqalign{
K^{(1)} _{2n+3}(u)&= \mathcal{R}^{(1)}(u)^n \partial_x^{-1}\delta_u (u_{2x}^{1/2}) \\
K^{(2)} _{2n+3}(u) &= \mathcal{R}^{(2)}(u)^n \partial_x^{-1}\delta_u (u_x^6u_{2x}^{-1}) \\
K^{(3)} _{2n+3}(u) &= \mathcal{R}^{(3)}(u)^n \partial_x^{-1}\delta_u (u_x^{-6}u_{2x}^{2})
}\quad (n=0,1,2,\cdots),
\end{equation*}
where the subscript $j$ in the symbol $K^{(i)}_j(u)$ indicates the highest order of $x$-derivatives of $u$, and usually is referred to as the order of symmetry. The leading (i.e. the third order) symmetry in each hierarchy was captured by investigating involution of some conserved densities under the Poisson bracket defined by $\partial_x^{-1}$, and in the same way, Wang conjectured \cite{wang} there would be another hierarchy of commuting symmetries leading by
\begin{equation*}
K^{(4)}_5(u) = \partial_x^{-1}\delta_u\left((u_{3x}-4u_x^{-1}u_{2x}^2)^{1/3}\right),
\end{equation*}
but the corresponding hereditary recursion operator is not known.

To get a clear picture of those symmetries and recursion operators listed above,  we will reformulate them in appropriate variables, and in such way we even establish an hereditary operator for $K^{(4)}_5$. This plan is supported by two notions: on the one hand, as image of the HS equation \eref{hs} under the reciprocal transformation $\mathcal{T}_0$, \Eref{hsode} admits abundant but more accessible  symmetries; on the other hand, the reciprocal transformation $\mathcal{T}_0$ could be extended to all $x$-independent symmetries of the HS equation \eref{hs:nonlocal} since $u_x^2$ produces the trivial $x$-translation symmetry as
\begin{equation*}
\partial_x^{-1}\delta_u(u_x^2/2) = u_x,
\end{equation*}
which qualifies $u_x^2$ as a common conserved density of all commuting flows corresponding to these symmetries.

To proceed, we first recall some relevant facts on symmetries and their behaviors under transformations \cite{fokas,fokas1}: {\it Suppose under an invertable change of variables
\begin{equation}\label{gcv}
	\left\{\eqalign{
		y = P(x,u,u_x,u_{2x},\cdots), \\
		v(y,t) = Q(x,u,u_x,u_{2x},\cdots),
	}\right.
\end{equation}
the equation $u_t = F(u,u_x,\cdots,u_{nx})$ is associated with $v_t = \tilde{F}(v,v_x,\cdots,v_{nx})$, then any symmetry of the former is mapped to that of the latter by the operator
\begin{equation}\label{opd}
\mathbf{D} = Q_u - (\partial_x P)^{-1}(\partial_x Q)P_u,
\end{equation}
where $Q_u$ ($P_u$ respectively) denotes the Frech\'{e}t derivative of $Q$ ($P$ respectively) with repect to $u$. }Let us  take the reciprocal transformation $\mathcal{T}_0$ as an example. It could be reformulated as
\begin{equation*}
\mathcal{T}_0 :\left\{\eqalign{
y &= \partial_x^{-1}u_x^2 \\
w &= u_x ,
}\right.
\end{equation*}
so symmetries of the HS equation are changed into those of \Eref{hsode} by the operator
\begin{equation*}
\eqalign{
\mathbf{D}_0 &= \partial_x - u_x^{-2}u_{2x}\partial_x^{-1}(2u_x)\partial_x = \partial_x + u_x\Big(\partial_x u_x^{-2} - u_x^{-2}\partial_x\Big)\partial_x^{-1}u_x\partial_x\\
&= u_x\partial_x u_x^{-2}\partial_x^{-1} u_x\partial_x = \left. w^3\partial_y w^{-2}\partial_y^{-1}w\partial_y\right|_{\partial_y = u_x^{-2}\partial_x,~w=u_x} .
}
\end{equation*}

\subsection{ $K^{(1)} _{3}(u)$, $K^{(2)} _{3}(u)$ and $K^{(3)} _{3}(u)$ in new variables}
With the Hamiltonian operator $\partial_x^{-1}$ and a conserved density $u_x^2G(r_2,r_3,\cdots,r_n)$, a class of symmetries of the HS equation \eref{hs:nonlocal},  given by
\begin{equation*}
K_{2n-1}(u) = \partial_x^{-1}\delta_u (u_x^2G(r_2,r_3,\cdots,r_n)),
\end{equation*}
 can be mapped to symmetries of \Eref{hsode} by the operator $\mathbf{D}_0$.  For instance, the third order symmetry
\begin{eqnarray*}
K_3(u) &= \partial_x^{-1}\delta_u (u_x^2G(r_2)) \\
&= \Big(u_x^{-6}u_{3x} - 4u_x^{-7}u_{2x}^2 \Big)G^{\prime\prime}(r_2) + 2u_x^{-3}u_{2x}G^{\prime}(r_2) -2u_x G(r_2)
\end{eqnarray*}
where the primes stand for derivatives of $G$ with respect to $r_2$, turns into
\begin{eqnarray*}
K_3(w) &= \mathbf{D}_0 K_3(u)|_{u_x=w,~\partial_x=w^2\partial_y} \\
&= \Big(w_{3y}-6w^{-1}w_yw_{2y}+6w^{-2}w_y^3\Big)G^{\prime\prime}(w^{-2}w_y) \\
&\quad+ \Big(w^{-2}w_{2y}^{2} - 4w^{-3}w_y^2w_{2y} + 4w^{-4}w_y^4\Big)G^{\prime\prime\prime}(w^{-2}w_y),
\end{eqnarray*}
which is a generic third order symmetry of \Eref{hsode}. Assuming $w=-v^{-1}$, \Eref{hsode} is converted to $v_t = 1$ and its symmetry $K_3(w)$ becomes
\begin{equation*}
K_3(v) = v^{2} K_3(w)|_{w=-v^{-1}}= v_{3y} G^{\prime\prime}(v_y) + v_{2y}^2G^{\prime\prime\prime}(v_y).
\end{equation*}
Thus, the precise form of symmetries  depends on $G$. Particularly taking appropriate $G$'s, for $K^{(1)} _{3}(u)$, $K^{(2)} _{3}(u)$ and $K^{(3)} _{3}(u)$,  our results are listed below.
\begin{enumerate}
	\item $G(\bullet) = (\bullet)^{1/2}$, and the corresponding symmetries:
	\begin{eqnarray*}
	K^{(1)} _{3}(u) = -\frac{1}{4}u_{2x}^{-3/2}u_{3x} , \\
	K^{(1)} _{3}(w) = -\frac{1}{4} w^3w_y^{-3/2}w_{3y} + \frac{3}{8}w^3 w_y^{-5/2} w_{2y}^2 ,\\
	K^{(1)} _{3}(v) = \frac{1}{4} v_y^{-3/2}v_{3y} - \frac{3}{8} v_y^{-5/2} v_{2y}^2 .
	\end{eqnarray*}
	\item $G(\bullet) = (\bullet)^{-1}$, and the corresponding symmetries:
	\begin{eqnarray*}
	K^{(2)} _{3}(u) = 2u_x^6 u_{2x}^{-3}u_{3x} - 12u_x^5u_{2x}^{-1} , \\
	K^{(2)} _{3}(w) = 2w^6 w_y^{-3} w_{3y} - 6w^6w_y^{-4}w_{2y}^2 + 12w^5w_y^{-2}w_{2y} -12 w^4 ,\\
	K^{(2)} _{3}(v) = 2v_y^{-3} v_{3y} -6v_y^{-4}v_{2y}^2 .
	\end{eqnarray*}
	\item $G(\bullet) = (\bullet)^{2}$, and the corresponding symmetries:
	\begin{eqnarray*}
	K^{(3)} _{3}(u) = 2u_x^{-6} u_{3x} - 6u_x^{-7}u_{2x}^{2} , \\
	K^{(3)} _{3}(w) = 2w_{3y} -12 w^{-1} w_y w_{2y} + 12w^{-2} w_y^3 ,\\
	K^{(3)} _{3}(v) = 2v_{3y} .
	\end{eqnarray*}
\end{enumerate}
The very last formula shows that the flow equation defined by $K^{(3)} _{3}(u)$  is a linear equation  in disguise. However, the flow equations associated with the other two, namely $K^{(1)} _{3}(u)$ and $K^{(2)} _{3}(u)$, still do not reveal themselves. To see their real faces, additional actions must be taken.

Regarding $K^{(1)} _{3}(v)$, a change of variables is defined as
\begin{equation*}
\left\{\eqalign{
z = \partial_y^{-1} v_y^{1/2} 	 \\
r= v_y,
}\right.
\end{equation*}
which plays the role of a reciprocal transformation\footnote{The flow equation of $K^{(1)} _{3}(v)$, i.e. $v_\tau = K^{(1)} _{3}(v)$ has a conservation law, i.e.
	\begin{equation*}
	\partial_\tau\Big(v_y^{1/2}\Big) = \partial_y \left(\frac{{1}}{8}v_y^{-2}v_{3y} -\frac{5}{32}v_y^{-3}v_{2y}^2\right),
	\end{equation*}
which implies a reciprocal transformation $(y,\tau,v)\rightarrow(z,\tau,r)$ defined by
\begin{equation*}
\textmd{d}z = v_y^{1/2}\textmd{d}y + \left(\frac{{1}}{8}v_y^{-2}v_{3y} -\frac{5}{32}v_y^{-3}v_{2y}^2\right)\textmd{d}t,\quad r(z,t)=v_y.
\end{equation*}}. Following the general changing rule of symmetries, $K^{(1)} _{3}(v)$ is mapped to
\begin{equation*}
\eqalign{
K^{(1)} _{3}(r) &=r^{3/2}\partial_z r^{-1/2}\partial_z^{-1}r^{-1/2}\partial_z  K^{(1)} _{3}(v)|_{v_y=r,~\partial_y = r^{1/2}\partial_z}\\
&= \frac{1}{4}r_{3z} -\frac{3}{4}r^{-1}r_zr_{2z} + \frac{15}{32}r^{-2}r_z^3.
}
\end{equation*}
Then assuming $p=\ln r$, we get
\begin{equation*}
K^{(1)} _{3}(p) =e^{-p} K^{(1)} _{3}(r)|_{r=e^p} = \frac{1}{4}p_{3z} - \frac{1}{32}p_z^3,
\end{equation*}
whose flow equation is nothing but the potential modified Korteweg-de Vries (KdV) equation
\begin{equation}\label{pmkdv}
p_\tau = \frac{1}{4}p_{3z} - \frac{1}{32}p_z^3.
\end{equation}

For $K^{(2)} _{3}(v)$, we introduce new variables as
\begin{equation*}
	\left\{\eqalign{
		z = v 	 \\
		r= v_y,
	}\right.
\end{equation*}
which should also be understood as a reciprocal transformation defined on a conservation law of the equation $v_\tau = K^{(2)} _{3}(v)$. After such transformation, $K^{(2)} _{3}(v)$ turns into
\begin{equation*}
K^{(2)} _{3}(r) =r^2\partial_z r^{-1} K^{(2)} _{3}(v)|_{v_y=r,~\partial_y=r\partial_z}= 2r_{3z} -12 r^{-1}r_zr_{2z} + 12 r^{-2}r_z^3,
\end{equation*}
which, after assuming $p=-r^{-1}$, is transformed to a linear symmetry
\begin{equation*}
K^{(2)} _{3}(p) =p^2 K^{(2)} _{3}(r)|_{r=-p^{-1}} = 2p_{3z}.
\end{equation*}

Under successive changes of variables, three leading symmetries, $K^{(1)}_3(u)$, $K^{(2)}_3(u)$ and $K^{(3)}_3(u)$, are respectively related to either linear or famous integrable equations, and the whole process is illustrated in the following diagram.
\begin{equation*}\fl
	\begin{CD}
	K^{(1)} _{3}(u)	@>{y=\partial_x^{-1} u_x^2}>{w=u_x}>K^{(1)} _{3}(w)	@> {v=-w^{-1}}>> K^{(1)} _{3}(v) @>{z=\partial_y^{-1} v_y^{1/2}}>{r=v_y}> K^{(1)} _{3}(r) @>{p=\ln r}>>K^{(1)} _{3}(p) \\
	K^{(2)} _{3}(u) @>{y=\partial_x^{-1} u_x^2}>{w=u_x}> K^{(2)} _{3}(w) @>{v=-w^{-1}}>> K^{(2)} _{3}(v) @>{z=v}>{r=v_y}> K^{(2)} _{3}(r) @>{p=-v^{-1}}>> K^{(2)} _{3}(p) \\
	K^{(3)} _{3}(u) @>{y=\partial_x^{-1} u_x^2}>{w=u_x}> K^{(3)} _{3}(w) @>{v=-w^{-1}}>> K^{(3)} _{3}(v)
	\end{CD}
\end{equation*}
Enlightened by these results,  we will consider three hierarchies $K^{(1)} _{2n+3}(u)$, $K^{(2)} _{2n+3}(u)$ and $K^{(3)} _{2n+3}(u)$ $(n=0,1,2,\cdots)$ by exploring their recursion operators under the transformations stated above.

\subsection{ $\mathcal{R}^{(1)}(u)$, $\mathcal{R}^{(2)}(u)$ and $\mathcal{R}^{(3)}(u)$ in new variables}
For convenience, we recall the behavior of recursion operators under changes of variables, which is also included in  \cite{fokas,fokas1}: {\it Suppose the equation $u_t = F(u,u_x,\cdots,u_{nx})$ is converted to $v_t = \tilde{F}(v,v_x,\cdots,v_{nx})$ by the transformation \eref{gcv}, then their recursion operators, $\mathcal{R}(u)$ and $\mathcal{R}(v)$, are related to each other via
\begin{equation*}
\mathcal{R}(v)|_{\partial_y = (\partial_x P)^{-1}\partial_x,~v=Q} = \mathbf{D}\mathcal{R}(u)\mathbf{D}^{-1},
\end{equation*}
where the operator $\mathbf{D}$ is defined by \Eref{opd}.
}

With  the help of above transformation rule of recursion operators, we may deduce  recursion operators corresponding to the variants of $K^{(1)} _{3}(u)$ from $\mathcal{R}^{(1)}(u)$, and they are given by
\begin{eqnarray*}
\mathcal{R}^{(1)}(w) & =\mathbf{D}_0\Big(\mathcal{R}^{(1)}(u)|_{u_x =w,~\partial_x = w^2\partial_y}\Big) \mathbf{D}_0^{-1} \\
%&= w^3\partial_y w^{-2} \partial_y^{-1} w \partial_y \Big(\mathcal{R}_1(u)|_{u_x =w,~\partial_x = w^2\partial_y}\Big) \partial_y^{-1}w^{-1}\partial_y w^2\partial_y^{-1} w^{-2} \\
&= 4 w^2w_y^{-1}\partial_y^2 - \Big(6w^2w_y^{-2}w_{2y}+4w\Big)\partial_y - 2w^2w_y^{-2}w_{3y} + 3 w^2w_y^{-3}w_{2y}^2\\
&\quad + 4ww_y^{-1}w_{2y} + \Big(w_yw_{3y}+\frac{3}{2}w_{2y}^2\Big)w^3w_y^{-5/2}\partial_y^{-1}w^{-1}w_y^{-3/2}w_{2y} \\
&\quad +\Big(2w_yw_{3y} - 3w_{2y}^2\Big)w^3w_y^{-5/2}\partial_y^{-1}w^{-2}w_y^{1/2},
\end{eqnarray*}
\begin{eqnarray*}
\mathcal{R}^{(1)}(v) &= v^2 \Big(\mathcal{R}^{(1)}(w)|_{w=-v^{-1}}\Big)v^{-2} \\
&= 4v_y^{-1}\partial_y^2 - 6v_y^{-2} v_{2y} \partial_y - 2v_y^{-2}v_{3y} + 3 v_y^{-3} v_{2y}^2 \\
&\quad - \Big(v_y^{-3/2} +\frac{3}{2}v_y^{-5/2}v_{2y}^2\Big)\partial_y^{-1} v_y^{-3/2}v_{2y},
\end{eqnarray*}
\begin{eqnarray*}
\mathcal{R}^{(1)}(r) &= r^{3/2}\partial_z r^{-1/2}\partial_z^{-1} r^{-1/2}\partial_z \Big(\mathcal{R}^{(1)}(v)|_{v_y=r,~\partial_y = r^{1/2}\partial_z}\Big) \partial_z^{-1} r^{1/2}\partial_z r^{1/2}\partial_z^{-1}r^{-3/2} \\
&= 4\partial_z^2 - 8r^{-1}r_z\partial_z - 4r^{-1}r_{2z} + 7r^{-2}r_z^2 + r_z\partial_z^{-1}\Big(r^{-2}r_{2z} - r^{-3}r_z^2\Big),
\end{eqnarray*}
and
\begin{eqnarray*}
\mathcal{R}^{(1)}(p) = e^{-p} \Big(\mathcal{R}^{(1)}(r)|_{r=e^p}\Big)e^p = 4\partial_z^2 - p_z^2 + p_z\partial_z^{-1} p_{2z},
\end{eqnarray*}
which is the recursion operator of $K^{(1)}_3(p)$, or equivalently, the potential modified KdV equation \eref{pmkdv}. Thererfore, {\it the hierarchy of symmetries $K^{(1)}_{2n+3}(u)(n=0,1,2,\cdots)$ of the HS equation has its roots in the potential modified KdV hierarchy.}

Along the route from $K^{(2)}_3(u)$ to $K^{(2)}_3(p)$, $\mathcal{R}^{(2)}(u)$ is step by step changed to
\begin{eqnarray*}
\mathcal{R}^{(2)}(w) & =\mathbf{D}_0\Big(\mathcal{R}^{(2)}(u)|_{u_x =w,~\partial_x = w^2\partial_y}\Big) \mathbf{D}_0^{-1} \\
&= 2w^4 w_y^{-2}\partial_y^2 -\Big(6w^4w_y^{-3}w_{2y} - 4w^3w_y^{-1}\Big)\partial_y - 2w^4w_y^{-3}w_{3y} \\
&\quad + 6w^4w_y^{-4}w_{2y}^2 - 4w^3w_y^{-2}w_{2y},
\end{eqnarray*}
\begin{eqnarray*}
\mathcal{R}^{(2)}(v) &= v^2 \Big(\mathcal{R}^{(2)}(w)|_{w=-v^{-1}}\Big)v^{-2} \\
&= 2v_y^{-2}\partial_y^2 - 6 v_y^{-3}v_{2y}\partial_y - 2v_y^{-3} v_{3y} + 6v_y^{-4}v_{2y}^2,
\end{eqnarray*}
\begin{eqnarray*}
\mathcal{R}^{(2)}(r) &= r^2\partial_z r^{-1}\Big(\mathcal{R}^{(2)}(v)|_{v_y=r,~\partial_y=r\partial_z}\Big) r\partial_z^{-1}r^{-2} \\
&= 2\partial_z^2 - 8r^{-1}r_z\partial_z - 4r^{-1}r_{2z} + 12r^{-2}r_z^2,
\end{eqnarray*}
and
\begin{equation*}
\mathcal{R}^{(2)}(p) = p^2\Big(\mathcal{R}^{(2)}(r)|_{r=-p^{-1}}\Big) p^{-2} = 2\partial_z^2,
\end{equation*}
which generates a hierarchy of linear symmetries
\[
K^{(2)}_{2n+3}(p)=\mathcal{R}^{(2)}(p)^nK^{(2)}_3(p)=2^{n+1}p_{(2n+3)z}, \quad (n=0,1,2,\cdots).
\]
Consequently, {\it any member in the hierarchy $K^{(2)}_{2n+3}(u)(n=0,1,2,\cdots)$ can be linearized through a series of changes  of variables.}

Following the tracks of $K^{(3)}_3(u)$, $\mathcal{R}^{(3)}(u)$ turns into
\begin{eqnarray*}
\mathcal{R}_3(w) & =\mathbf{D}_0\Big(\mathcal{R}_3(u)|_{u_x =w,~\partial_x = w^2\partial_y}\Big) \mathbf{D}_0^{-1} \\
&= 2\partial_y^2 - 8w^{-1}w_y\partial_y -4 w^{-1}w_{2y} + 12w^{-2}w_y^2,
\end{eqnarray*}
and
\begin{equation*}
\mathcal{R}_3(v) = v^2 \Big(\mathcal{R}_3(w)|_{w=-v^{-1}}\Big)v^{-2} = 2\partial_y^2,
\end{equation*}
which generates a hierarchy of linear symmetries
\[
K^{(3)}_{2n+3}(v)=\mathcal{R}^{(3)}(v)^nK^{(3)}_3(v)=2^{n+1}v_{(2n+3)z},\quad (n=0,1,2,\cdots).
\]
Hence,  {\it the hierarchy $K^{(3)}_{2n+3}(u)(n=0,1,2,\cdots)$ is also linearized through even less changes of variables. }

{\bf Remarks:} These recursion operators are all figured out with the help of SUSY2, a package \cite{susy2} designed by Popowicz for the computer algebra system Reduce. The calculations involved require certain manipulation of the  expressions with multiple integrating operators, like $\partial_y^{-1} w^kw_y\partial_y^{-1}$, $\partial_y^{-1} \Big(2w^k w_y w_{2y} + kw^{k-1}w_y^3\Big)\partial_y^{-1}$ and $\partial_y^{-1} \Big(2w^kw_{3y} - k(k-1)w^{k-2}w_y^3\Big)\partial_y^{-1}$, and indeed the multiplicity of integrating operators could be reduced through integration by parts, for instance
\begin{eqnarray*}
\partial_y^{-1} w^kw_y\partial_y^{-1} = \frac{1}{k+1}\Big(w^{k+1}\partial_y^{-1} - \partial_y^{-1}w^{k+1}\Big),\quad (k\neq -1)\\
\partial_y^{-1} \Big(2w^k w_y w_{2y} + kw^{k-1}w_y^3\Big)\partial_y^{-1} = w^k w_y^2\partial_y^{-1} - \partial_y^{-1} w^k w_y^2,\\
\partial_y^{-1} \Big(2w^kw_{3y} - k(k-1)w^{k-2}w_y^3\Big)\partial_y^{-1} \\
~~~~~~~~~~ = \left(w^kw_{2y} - \frac{k}{2}w^{k-1}w_y^2\right)\partial_y^{-1} - \partial_y^{-1}\left(w^kw_{2y} - \frac{k}{2}w^{k-1}w_y^2\right).
\end{eqnarray*}

\subsection{Recursion operator of $K^{(4)}_5(u)$}
As mentioned before, there is supposed to be another hierarchy of commuting symmetries leading by $K^{(4)}_5(u)$, but its recursion operator is still absent. In this subsection, we manager  to establish its recursion operator via the link between $K^{(4)}_5(u)$ and the Fordy-Gibbons equation \cite{fordy}.

Let us demonstrate the process of relating $K^{(4)}_5(u)$ with the Fordy-Gibbons equation.  Firstly applying the reciprocal transformation $\mathcal{T}_0$ to $K^{(4)}_5(u)$,
%\begin{eqnarray*}
%K^{(4)}_5(u) = \partial_x^{-1}\delta_u\left((u_{3x}-4u_x^{-1}u_{2x}^2)^{1/3}\right),
%\end{eqnarray*}
we have
\begin{eqnarray*}
K^{(4)}_5(w) &=\mathbf{D}_0 K^{(4)}_5(u)_{u_x=w,~\partial_x=w^2\partial_y}.
\end{eqnarray*}
Secondly assuming $v=-w^{-1}$, we get
\begin{eqnarray*}
K^{(4)}_5(v) = v^{2} K^{(4)}_5(w)|_{w=-v^{-1}} = \frac{2}{9}v_{2y}^{5/3}v_{5y} - \frac{10}{9}v_{2y}^{8/3}v_{3y}v_{4y} + \frac{80}{81}v_{2y}^{11/3}v_{3y}^3.
\end{eqnarray*}
Thirdly introducing a reciprocal transformation
\begin{equation*}
\left\{\eqalign{
z = \partial_y^{-1}v_{2y}^{1/3}\\
r = v_{2y},
}\right.
\end{equation*}
and applying it to $K^{(4)}_5(v) $ yields
\begin{eqnarray*}
K^{(4)}_5(r) &= r^{4/3}\partial_z r^{-1/3}\partial_z^{-1} r^{-2/3}\partial_z r^{1/3}\partial_z K^{(4)}_5(r)|_{v_{2y=r,~\partial_y=r^{1/3}\partial_z}}\\
&=\frac{5}{9}r_{5z} - \frac{10}{9}r^{-1}r_zr_{4z} - \frac{70}{27}r^{-1}r_{2z}r_{3z} + \frac{380}{81}r^{-2}r_z^2r_{3z} + \frac{620}{81}r^{-2}r_zr_{2z}^2 \\
&\quad -\frac{1180}{81}r^{-3}r_z^3r_{2z} + \frac{4160}{729}r^{-4}r_z^5.
\end{eqnarray*}
Finally let $p=r^{-1}r_z/3$, then we obtain
\begin{eqnarray*}
K^{(4)}_5(p) &=\frac{1}{3}\partial_z \exp{(-3\partial_z^{-1}p)}  K^{(4)}_5(r)|_{r=\exp{(3\partial_z^{-1}p)}}\\
&= \frac{2}{9}\Big(p_{5z} - 5p_zp_{3z} - 5p^2p_{3z} -5p_{2z}^2 - 20pp_zp_{2z} - 5p_z^3 + 5p_zp^4\Big),
\end{eqnarray*}
whose flow equation is just the Fordy-Gibbons equation \cite{fordy}, also known as a common modification of the Sawada-Kotera equation and the Kaup-Kupershmidt equation.

From the recursion operator of the Fordy-Gibbons equation, deduced from that of the Sawada-Kotera (or Kuper-Kupershmidt) equation via a Miura-type transformation,
\begin{eqnarray*}
\mathcal{R}^{(4)}(p) = \partial_z(\partial_z+2p)\partial_z^{-1}(\partial_z-p)(\partial_z+p)\partial_z(\partial_z-p)(\partial_z+p)\partial_z^{-1}(\partial_z-2p),
\end{eqnarray*}
we may derive the recursion operator of $K^{(4)}_5(u)$ in different coordinates as follows
\begin{eqnarray*}
\mathcal{R}^{(4)}(r) &=  r\partial_z^{-1} \Big(\mathcal{R}^{(4)}(p)|_{p=r^{-1}r_z/3}\Big)\partial_z r^{-1} \\
&= r^{1/3}\partial_z r^{2/3}\partial_z^{-1}r^{1/3}\partial_z r^{-2/3}\partial_z r^{1/3} \partial_z r^{1/3}\partial_z r^{-2/3}\partial_z r^{1/3} \\
&\qquad \cdot \partial_z^{-1} r^{2/3} \partial_z r^{-2/3}\partial_z r^{-1},
\end{eqnarray*}
\begin{eqnarray*}
\mathcal{R}^{(4)}(v) &= \partial_y^{-2} v_{2y}^{2/3}\partial_yv_{2y}^{1/3}\partial_y^{-1} v_{2y}^{-1} \Big(\mathcal{R}^{(4)}(r)|_{r=v_{2y},~\partial_z=v_{2y}^{-1/3}\partial_y}\Big) \\
&\qquad\cdot v_{2y}\partial_y v_{2y}^{-1/3}\partial_y^{-1}v_{2y}^{-2/3}\partial_y^2 \\
&= \partial_y^{-2} v_{2y}^{2/3}\partial_yv_{2y}^{1/3}\partial_y^{-1} v_{2y}^{-1}\partial_y v_{2y}^{2/3} \partial_y^{-1} v_{2y}^{1/3}\partial_y v_{2y}^{-1} \partial_y^3 v_{2y}^{-1}\partial_yv_{2y}^{1/3} \\
&\qquad\cdot \partial_y^{-1}v_{2y}^{2/3}\partial_y v_{2y}^{-1}\partial_y^2 v_{2y}^{-1/3}\partial_y^{-1}v_{2y}^{-2/3}\partial_y^2,
\end{eqnarray*}
\begin{eqnarray*}
\mathcal{R}^{(4)}(w) &= w^2\Big(\mathcal{R}^{(4)}(v)|_{v=-w^{-1}}\Big)w^{-2},
\end{eqnarray*}
and
\begin{eqnarray*}
\mathcal{R}^{(4)}(u) &= \partial_x^{-1}u_x^{-1}\partial_xu_x^{2}\partial_x^{-1} u_x^{-1} \Big(\mathcal{R}^{(4)}(w)|_{w=u_x,~\partial_y= u_x^{-2}\partial_x}\Big)u_x\partial_xu_x^{-2}\partial_x^{-1}u_x\partial_x \\
&= \partial_x^{-1}u_x^{-1}\partial_xu_x^{2}\partial_x^{-1} u_x \Big(\mathcal{R}^{(4)}(v)|_{v=-u_x^{-1},~\partial_y= u_x^{-2}\partial_x}\Big)u_x^{-1}\partial_xu_x^{-2}\partial_x^{-1}u_x\partial_x \\
&= \partial_x^{-1}u_x^{-1}\partial_xu_x^{2}\partial_x^{-1} u_x \partial_x^{-1} u_x^2\partial_x^{-1} r_3^{2/3} \partial_x r_3^{1/3} \partial_x^{-1}r_3^{-1} \partial_x r_3^{2/3} \partial_x^{-1} r_3^{1/3}\partial_x r_3^{-1} \\
&\quad \cdot (u_x^{-1}\partial_x)^3 r_3^{-1}u_x^{-2} \partial_x r_3^{1/3} \partial_x^{-1} r_3^{2/3} \partial_x r_3^{-1} (u_x^{-2}\partial_x)^2 r_3^{-1/3}\partial_x^{-1}r_3^{-2/3}\partial_x u_x^{-1} \\
&\quad \cdot \partial_x u_x^{-1}\partial_xu_x^{-2}\partial_x^{-1}u_x\partial_x
\end{eqnarray*}
where $r_3$ is introduced by Lemma \ref{lem:1}. $\mathcal{R}^{(4)}(u)$ in present form is rather involved and it would be interesting to find a simpler version for it. Unfortunately, we  failed to do so due to the fractional powers of $r_3$.

\section{Conclusions}
On the basis of a simple conservation law, a reciprocal transformation $\mathcal{T}_0$ was defined and adopted to convert the HS equation \eref{hs} to \Eref{hsode}, whose symmetries is easily accessible  due to its equivalence to $v_t=1$, and enabled us to unveil  remarkable properties of the HS equation from a new angle. Through the inverse of $\mathcal{T}_0$, a class of conserved densities for the HS equation were easily  constructed from infinitesimal symmetries of \Eref{hsode}, and expressed as arbitrary functions of basic variables $r_k(k=2,3,\cdots)$, which are images of  $v_k$, invariants of $v_t=1$, under the inverse of $\mathcal{T}_0$. In addition, conserved densities satisfying \eref{law:2} and even nonlocal ones were also presented.

Benefited from involution of the conserved density $u_x^2$ with others, $\mathcal{T}_0$ was applied to symmetries of the HS equation, and allowed us to essentially get a better understanding of them. Through appropriate changes of variables beginning with $\mathcal{T}_0$, three hierarchies of commuting symmetries of the HS equations were either converted to symmetries belonging to the potential modified KdV hierarchy, or even linearized. Following the same strategy, we constructed a hereditary recursion operator for a fifth order symmetry $K^{(4)}_5(u)$.

This paper reasonably explains exotic features discovered by Wang \cite{wang} on the one hand, and launches some new questions about the HS equation on the other hand. For instance, are there other hierarchies of commuting symmetries leading by some special $K_3(u)$? What will happen if one studies commutativity of symmetries generated by conserved densities given by Proposition \ref{pro:2}? The most interesting and challenging question may be how to completely find out conserved densities and symmetries for the HS equation.

\section*{Acknowledgment}
Both authors would like to thank Dr. Jing Ping Wang for her interests in this work and valuable discussions. This work is supported by National Natural Science Foundation of China (NNSFC) (Grant Nos. 11271366 and 11331008), Specialized Research Fund for the Doctoral Program of Higher Education (SRFDP) (Grant No. 20120023120006) and Fundamental Research Funds for Central Universities (Grant No. 2011QS02).

\appendix

\section{Nonlocal conserved densities of the HS equation}
Since the HS equation has a conservation law
\begin{equation*}
\partial_t u_x^2 = \partial_x(2uu_x^2),
\end{equation*}
we introduce a potential $v$ such that
\begin{equation}\label{pot}
\left\{\eqalign{
v_x &= u_x^2 \\
v_t &= 2uu_x^2,
}\right.
\end{equation}
and rewrite the evolution form \eref{hs:nonlocal} of the HS equation as
\begin{equation}\label{hs:nonlocal:2}
u_t = 2uu_x -v .
\end{equation}
With the potential $v$ and $r_k$'s defined in Lemma \ref{lem:1}, we have
\begin{lem}\label{lem:2}
Let $R_k = v^2r_k~(k=1,2,\cdots)$, then
\begin{equation*}
\partial_t R_1 = 2u(\partial_x R_1) + v^2\quad \mbox{and} \quad \partial_tR_k = 2u(\partial_x R_k)\quad (k=2,3,\cdots).
\end{equation*}
\end{lem}
\begin{proof}
Direct calculation.
\end{proof}

\begin{prop}
For the HS equation \eref{hs}, $F(v,u,u_x,u_{2x},\cdots,u_{nx})$ serves as a conserved density such that
\begin{equation}\label{law:3}
\partial_t F(v,u,u_x,u_{2x},\cdots,u_{nx}) = \partial_x\Big(2uF(v,u,u_x,u_{2x},\cdots,u_{nx})\Big).
\end{equation}
if and only if
\begin{equation*}
F(v,u,u_x,u_{2x},\cdots,u_{nx}) = v^{-4}u_x^2 G(v,vu-v^2u_x^{-1},R_2,R_3,\cdots,R_n),
\end{equation*}
where $G$ is an arbitrary smooth function of its arguments.
\end{prop}
\begin{proof}
The sufficiency is proved by direct calculation, while the necessity will be shown by solving all conserved densities in \Eref{law:3}.

Through an invertible change of coordinates
\begin{equation*}
\bar{\Gamma}: (v,u,u_x,u_{2x},\cdots,u_{nx})\mapsto(v,u,R_1,R_2,\cdots,R_n),
\end{equation*}
which is well-defined since $\partial R_k/\partial u_{kx} = v^2u_x^{-2k}~(k=1,2,\cdots)$, the conserved  density $F(v,u,u_x,u_{2x},\cdots,u_{nx})$ may be reformulated as  $\hat{F}(v,u,R_1,R_2,\cdots,R_n)$. Then the conservation law \eref{law:3} is calculated as
\begin{equation*}
\eqalign{
0 &= \partial_t \hat{F}(v,u,R_1,R_2,\cdots,R_n) - \partial_x\Big(2u\hat{F}(v,u,R_1,R_2,\cdots,R_n)\Big) \\
&= v_t\frac{\partial \hat{F}}{\partial v} + u_t \frac{\partial \hat{F}}{\partial u} + \sum_{k=1}^n (\partial_t R_k)\frac{\partial \hat{F}}{\partial R_k} - 2u v_x \frac{\partial \hat{F}}{\partial v} - 2uu_x\frac{\partial \hat{F}}{\partial u} \\
& \quad - 2u\sum_{k=1}^n(\partial_x R_k)\frac{\partial \hat{F}}{\partial R_k} - 2u_x \hat{F} .
%&= 2uu_x^2  \frac{\partial \hat{F}}{\partial v} + (2uu_x - v) \frac{\partial \hat{F}}{\partial u} + \sum_{k=1}^n 2u(\partial_x R_k)\frac{\partial \hat{F}}{\partial R_k} + v^2\frac{\partial \hat{F}}{\partial R_1} \\
%&\quad - 2u u_x^2 \frac{\partial \hat{F}}{\partial v} - 2uu_x\frac{\partial \hat{F}}{\partial u} - 2u\sum_{k=1}^n(\partial_x R_k)\frac{\partial \hat{F}}{\partial R_k} - 2u_x \hat{F} \\
%&= -v\frac{\partial \hat{F}}{\partial u} + v^2\frac{\partial \hat{F}}{\partial R_1} - 2u_x \hat{F} \\
%&= -v\frac{\partial \hat{F}}{\partial u} + v^2\frac{\partial \hat{F}}{\partial R_1} + 2 v^2R_1^{-1} \hat{F}
}
\end{equation*}
Replacing $v_x$, $v_t$, $u_t$ and $(\partial_t R_k)$ by Equations \eref{pot}, \eref{hs:nonlocal:2} and those in Lemma \ref{lem:2}, then simplifying these expressions yields
\begin{equation*}
0 = -v\frac{\partial \hat{F}}{\partial u} + v^2\frac{\partial \hat{F}}{\partial R_1} - 2u_x \hat{F} = -v\frac{\partial \hat{F}}{\partial u} + v^2\frac{\partial \hat{F}}{\partial R_1} + 2 v^2R_1^{-1} \hat{F}.
\end{equation*}
Solving it by the characteristic method, we obtain the conserved density in coordinates $(v,u,R_1,R_2,\cdots,R_n)$
\begin{equation*}
\hat{F}(v,u,R_1,R_2,\cdots,R_n) = R_1^{-2} G(v,vu+R_1,R_2,R_3,\cdots,R_n),
\end{equation*}
where $G$ is an arbitrary smooth function.
\end{proof}

\end{document}